\documentclass[11pt]{article}

\usepackage[a4paper,margin=1in]{geometry}
\usepackage{setspace}
\usepackage{parskip}    

\usepackage[T1]{fontenc}
\usepackage{lmodern}
\usepackage{amsmath,amssymb,amsthm}
\usepackage{bm}
\usepackage{mathtools}
\usepackage{stmaryrd}   
\usepackage{mathpartir} 
\usepackage{enumitem}
\usepackage{microtype}
\usepackage{xcolor}
\usepackage{tikz}
\usepackage{thmtools}   
\usepackage[colorlinks=true,allcolors=blue]{hyperref}
\usepackage[capitalise,noabbrev]{cleveref}  
\newcommand{\rocqrepo}{https://github.com/park-sewon/Diaconescu/blob/arxiv}

\newcommand{\githublink}[3][\rocqfile]{\href{\rocqrepo/#1\#L#2}{\texttt{#3}}}

\newcommand{\tytype}{\mathsf{Type}}
\newcommand{\typrop}{\mathsf{Prop}}
\newcommand{\El}{\mathrm{El}}
\newcommand{\Id}{\mathsf{Id}}
\newcommand{\refl}{\mathsf{refl}}
\newcommand{\W}{\mathsf{W}}
\newcommand{\rec}{\mathsf{rec}}
\newcommand{\lift}{\mathsf{lift}}
\newcommand{\Quot}{\mathsf{Quot}}
\newcommand{\mk}{\mathsf{mk}}
\newcommand{\qlift}{\mathsf{lift}}

\newcommand{\choice}{\mathsf{choice}}
\newcommand{\supp}[1]{\lVert #1 \rVert}

\newcommand{\ex}[1]{\exists #1.\,}   
\newcommand{\cst}{\bm{\Delta}}  
\newcommand{\pproj}{\mathsf{p}}  
\newcommand{\qvar}{\mathsf{q}}   
\newcommand{\Wsup}{\mathsf{sup}} 
\newcommand{\pit}[1]{\Pi_{#1}\,}  
\newcommand{\all}[1]{\forall #1.\,}  
\newcommand{\defto}{\mathrel{:\equiv}}  
\newcommand{\of}{{:}}   
\newcommand{\lto}{\Rightarrow}       
\newcommand{\liff}{\Leftrightarrow}  

\newcommand{\Lo}{L}   
\newcommand{\decode}[1]{\llcorner #1 \lrcorner}
\newcommand{\sset}{\mathbf{Stg}}
\newcommand{\SPi}{\bm{\Pi}}        
\newcommand{\SSigma}{\bm{\Sigma}}  

\newcommand{\functy}{\mathsf{Ty}}
\newcommand{\functm}{\mathsf{Tm}}
\newcommand{\Ty}{\functy}
\newcommand{\Tm}{\functm}  

\DeclareUnicodeCharacter{2080}{\textsubscript{0}}

\declaretheorem[name=Theorem, style=plain, numberwithin=section]{thm}
\declaretheorem[name=Lemma,       style=plain, sibling=thm]{lem}

\declaretheorem[name=Definition,  style=definition, sibling=thm]{defi}

\declaretheorem[name=Remark,      style=remark, sibling=thm]{rem}

\numberwithin{equation}{section}

\makeatletter
\@for\@tempa:=thm,lem,prop,cor,conj,defi,exa,nota,rem\do{%
  \expandafter\xdef\csname theH\@tempa\endcsname{\noexpand\thesection.\noexpand\arabic{thm}}}
\makeatother

\title{Global Choice without Excluded Middle \\in Dependent Type Theory}
\hypersetup{pdftitle={Global Choice without Excluded Middle in
Dependent Type Theory}}

\author{
  Sewon Park\\
  \small Institute of Mathematics, Physics and Mechanics, Ljubljana, Slovenia\\
  \small\texttt{sewon@sewonpark.com}
}

\date{}

\begin{document}

\maketitle

\begin{abstract}
Diaconescu's theorem derives excluded middle from the axiom of choice, using propositional extensionality.
We show that in dependent type theory without propositional extensionality, global choice alone does not imply excluded middle.
We consider extensional Martin-Löf type theory with an impredicative universe of propositions included in the universe of types, extended with a global choice operator.
We construct a consistent model of this type theory and show that the model refutes excluded middle and propositional extensionality.
The model is built from staged sets, sets whose elements appear at stages of a well-order.
It further validates proof irrelevance, inductive families, quotients with soundness, and a cumulative hierarchy of universes, which together cover a large fragment of the type theory of the proof assistant Lean.
Even together with these features, global choice does not imply excluded middle.

\end{abstract}


\section{Introduction}\label{sec:intro}

Diaconescu's theorem states that the axiom of choice implies the law of excluded middle in every topos~\cite{Dia75}.
The proof uses propositional extensionality, the principle that equivalent propositions are equal.
Propositional extensionality holds in all toposes, but in dependent type theory it is an optional axiom.
For example, the core library of Lean~\cite{dMKAvDvR15} assumes propositional extensionality and uses it to derive excluded middle from global choice~\cite{LeanClassical}.
We show that in dependent type theory global choice alone does not imply excluded middle.

We consider extensional Martin-Löf type theory~\cite{ML84} with an impredicative universe $\typrop$ of propositions.
The universe $\typrop$ is included in the universe $\tytype$ of types, with identity types $\Id_A(a, b) : \typrop$ for a type $A$ and $a, b : A$.
In this type theory, global choice is an operator
\begin{equation*}
  \choice : \pit{X\of{}\tytype} \supp{X} \to X\;,
\end{equation*}
where $\supp{X} :\equiv \all{P\of{}\typrop} (X \to P) \to P$ is the propositional truncation of~$X$, encoded using the impredicativity of~$\typrop$.

\begin{thm}\label{thm:main}
The type theory above with a global choice operator is consistent with
\begin{equation*}
  \neg\,\all{P\of{}\typrop} P \vee \neg P
  \qquad\text{and}\qquad
  \neg\,\all{P,Q\of{}\typrop} (P \liff Q) \lto \Id_\typrop(P, Q)\;.
\end{equation*}
\end{thm}

We prove \cref{thm:main} with a model built from staged sets.
We first describe the model informally.
\Cref{sec:staged,sec:model,sec:logic} give the precise definitions.
A staged set is a set whose elements appear at stages $0, 1, 2, \dots$.
We interpret a proposition as the stage at which it becomes true, or as $\infty$ if it never becomes true.
Thus two propositions $p$ and $q$ that become true at different stages are equivalent from stage $\max(p, q)$ on.
The universe $\typrop$ is the staged set of the stages together with $\infty$, each appearing at stage $0$.
Hence the identity proposition $\Id_\typrop(p, q)$ compares the stages themselves:
it holds from stage $0$ on when $p$ and $q$ coincide, and it never holds otherwise.
For example, for $p :\equiv 0$ and $q :\equiv 1$, the equivalence $p \liff q$ holds from stage $1$ on, while $\Id_\typrop(p, q)$ never holds.
Hence propositional extensionality fails in the model.

Excluded middle fails for a similar reason.
For a proposition $p$ that becomes true, the disjunction $p \vee \neg p$ holds only from stage $p$ on.
To prove $\all{p\of{}\typrop} p \vee \neg p$ at a stage $n$, we assume $p \of{} \typrop$.
Since every proposition in $\typrop$ appears at stage $0$, assuming $p$ does not increase the stage.
Hence $p \vee \neg p$ has to hold at stage $n$ for every $p$.
When there is no last stage, this fails at every stage $n$.

Types are staged sets as well.
The truncation $\supp{X}$ of a type $X$ becomes true at the stage at which the first element of $X$ appears.
The choice operator maps each nonempty type to one of the elements that appear first, chosen in the classical metatheory.
Hence the chosen element is available as soon as $\supp{X}$ becomes true.

\Cref{thm:main} applies to intensional Martin-Löf type theory as well, because the rules of the intensional theory are derivable in the extensional one.
Since the model is a counter-model, every further rule that it validates strengthens the result.
The model also validates proof irrelevance, inductive families, quotients with soundness, and a cumulative hierarchy of universes that contain $\typrop$.
That means global choice does not imply excluded middle even together with these rules.
These rules cover a large fragment of the type theory of Lean, as specified in Carneiro's thesis~\cite{Carn19}.

\Cref{sec:staged} introduces staged sets.
\Cref{sec:model} lists the extra rules of the type theory for completeness, and then constructs the model as a category with families (CwF)~\cite{Hof97} of families of staged sets.
\Cref{sec:logic} computes the logic of the model and proves \cref{thm:main}.
\Cref{sec:more} shows further features that the model validates and relates them to Lean's type theory.
\Cref{sec:related} discusses related work.
\Cref{app:proofs} lists the proof obligations of the model, as verified in Rocq~\cite{Rocq}.

\section{Staged sets}\label{sec:staged}
Fix a nonempty well-ordered set $\Lo$.
We write $0$ for the least element of $\Lo$.
\begin{defi}
  A \emph{staged set} is a set $X$ together with a function $\nu : X \to \Lo$.
  A \emph{staged map} $f : (X, \nu) \to (Y, \mu)$ is a function $f : X \to Y$ such that $\mu(f(x)) \leq \nu(x)$ for all $x \in X$.
\end{defi}
We call $\nu(x)$ the \emph{birthday} of $x$, and say that $x$ is \emph{available} at a stage $\ell$ if $\nu(x) \le \ell$.
A staged map thus sends each element to one that is available when the element is born.

For a staged set $X$, we write $|X|$ for the underlying set and $\nu_X$ for the function $\nu$. When there is no ambiguity, we write $x \in X$ for $x \in |X|$ and $\nu(x)$ for $\nu_X(x)$.
For a set $X$, let $\cst X \defto (X, x \mapsto 0)$.

A \emph{family} $Y$ \emph{of staged sets} over a staged set $X$ assigns a staged set $Y(x)$ to each $x \in |X|$.

For a family $Y$ over $X$:
\begin{enumerate}
  \item the \emph{staged sum} $\SSigma(X, Y)$ is the staged set with
  \begin{equation*}
  \lvert \SSigma(X, Y)\rvert \defto \{(x, y) \mid x \in |X| \land y \in |Y(x)|\}
  \ \text{and}\ %
  \nu_{\SSigma(X, Y)}(x, y) \defto \max(\nu_X(x), \nu_{Y(x)}(y))\;;
  \end{equation*}
  \item the \emph{staged product} $\SPi(X, Y)$ is the staged set whose underlying set consists of the functions $f$ with $f(x) \in |Y(x)|$ for all $x \in |X|$ such that, for some $\ell \in \Lo$,
  \begin{equation*}
  \nu_{Y(x)}(f(x)) \leq \max(\nu_X(x), \ell) \quad\text{for all } x \in |X|\;,
  \end{equation*}
  and $\nu_{\SPi(X, Y)}(f)$ is the least such $\ell$, which exists since $\Lo$ is well-ordered;
  \item a \emph{section} is an element of $\SPi(X, Y)$ whose birthday is $0$.
\end{enumerate}
In words, the staged set $\SPi(X, Y)$ consists of the set-theoretic dependent functions $f$ that map each $x \in |X|$ to an element of $|Y(x)|$ that is available once both $x$ and $f$ are.
Not all set-theoretic dependent functions are included in $\SPi(X, Y)$.
For example, for $\Lo = \mathbb{N}$, $X = \cst\mathbb{N}$ and the constant family $Y$ with value $(\mathbb{N}, m \mapsto m)$, the function $n \mapsto n$ sends each $n$, born at $0$, to $n$, born at $n$.
These birthdays are unbounded, and hence $n \mapsto n$ has no birthday and is not an element of $\SPi(X, Y)$.

Staged sets are, up to equivalence, the presheaves on $\Lo^{\mathrm{op}}$ with injective restriction maps.
The presheaf of a staged set $X$ maps a stage $\ell$ to the set $\{x \in |X| \mid \nu(x) \le \ell\}$ of elements available at $\ell$.
Since $\Lo$ is a chain, these are the $\neg\neg$-separated presheaves, and staged sets form a quasitopos~\cite[Theorem~A4.4.5]{Johnstone:592033}.
(Under this equivalence, $\cst X$ is the constant presheaf with value $X$, which explains the notation.)
However, we interpret the type theory with the concrete description of staged sets, instead of the structure of the quasitopos.
First, staged sets are simple to calculate with directly.
Second, the model differs from the structure of the quasitopos: $\typrop$ is not the strong subobject classifier $\cst\{0, 1\}$, which would validate excluded middle, and type families are not objects of the slice category, as \cref{sec:why-precomposition} explains.

\subsection{Truth values}
Let $\Lo_\infty \defto \Lo \cup \{\infty\}$, with $\infty$ as the top element.
For $p \in \Lo_\infty$, let
\begin{equation*}
\decode{p} \defto
\begin{cases}
(\{\ast\},\ \ast \mapsto p) & \text{if } p \in \Lo\;,\\
\emptyset & \text{if } p = \infty\;.
\end{cases}
\end{equation*}
Up to isomorphism, every staged set with at most one element is $\decode{p}$ for a unique $p \in \Lo_\infty$. These are all the subobjects of the terminal staged set $\decode{0}$.
We read $p$ as a truth value that holds from stage $p$ on, that is, at every stage $\ell \ge p$.
Thus $0$ is true from the start, $\infty$ is never true, and a smaller truth value is truer.
Since $\Lo$ is well-ordered, every subset of $\Lo_\infty$ has a supremum.

The \emph{birthday} of a staged set $X$ is the earliest birthday of the elements of $X$,
\begin{equation*}
\nu(X) \defto \min\{\nu(x) \mid x \in |X|\} \in \Lo_\infty\;,
\end{equation*}
with $\min \emptyset = \infty$.

\section{The model}\label{sec:model}
Hofmann~\cite{Hof97} explains in detail how a CwF interprets a type theory and works out the conditions for several type formers.
We give the constructions of the CwF and the non-routine verifications.
\Cref{app:proofs} lists the remaining verifications.

\subsection{The extra rules}
The type theory of \cref{thm:main} extends the rules of~\cite{ML84} as follows.
The universes $\tytype$ and $\typrop$ are presented à la Tarski, with decoding $\El$.
The universe $\typrop$ has a code in $\tytype$, and an inclusion
\begin{equation*}
\iota : \typrop \to \tytype\;, \qquad \El(\iota(P)) = \El(P) \quad \text{for } P : \typrop\;.
\end{equation*}
The universe $\typrop$ is impredicative: for every type $A$ and every $P : A \to \typrop$,
\begin{equation*}
\all{x\of{}A} P\,x : \typrop\;,
\end{equation*}
with introduction and elimination rules as for $\Pi$-types, including the judgmental $\beta$- and $\eta$-rules.
For a type $A$ and $a, b : A$, the identity proposition
\begin{equation*}
\Id_A(a, b) : \typrop
\end{equation*}
decodes to the identity type of~\cite{ML84}, with equality reflection.
The global choice operator has the type given in \cref{sec:intro}.
We interpret global choice in \cref{sec:logic}, where we also refute excluded middle and propositional extensionality.

\subsection{Category with families}

The CwF of families of staged sets has the following structure.
\begin{enumerate}
  \item A \emph{context} $\Gamma$ is a staged set, and a context morphism is a staged map.
  The \emph{empty context} is the terminal staged set $\decode{0}$.
  \item A \emph{type} over $\Gamma$ is a family of staged sets over $\Gamma$.
  \item A \emph{term} of type $A$ over $\Gamma$ is a section of the family $A$.
  \item The \emph{context extension} of $\Gamma$ by a type $A$ over $\Gamma$
  is the staged sum $\Gamma.A \defto \SSigma(\Gamma, A)$, with the \emph{projection} $\pproj(\gamma, a) \defto \gamma$ and the \emph{variable} $\qvar(\gamma, a) \defto a$.
  \item The \emph{substitution} along $\sigma : \Delta \to \Gamma$ of a type $A$ over $\Gamma$ is $A[\sigma] \defto A \circ \sigma$,
  and that of a term $t$ of type $A$ is $t[\sigma] \defto t \circ \sigma$.
\end{enumerate}

\subsection{Type formers}

\paragraph{Universes.}
Fix an inaccessible cardinal $\kappa$, and let $V_\kappa$ be the set of sets of rank below $\kappa$, a Grothendieck universe.
We assume $\Lo \in V_\kappa$.
We call a staged set $X$ \emph{$\kappa$-small} if $|X| \in V_\kappa$, and write $\sset_\kappa$ for the set of $\kappa$-small staged sets.
A family $Y$ over $X$ is \emph{$\kappa$-small} if every $Y(x)$ is $\kappa$-small, that is, if $Y$ is a function $|X| \to \sset_\kappa$.
Over a context $\Gamma$, the universe $\tytype$ is the constant family $\gamma \mapsto \cst \sset_\kappa$.
Since every element of $\cst \sset_\kappa$ is born at $0$, a term of type $\tytype$ is a function $X : |\Gamma| \to \sset_\kappa$, that is, a $\kappa$-small family over $\Gamma$, and we put
\begin{equation*}
\El(X) \defto X\;.
\end{equation*}
The decoding $\El$ commutes with substitution because both $\El(X)[\sigma]$ and $\El(X[\sigma])$ are $X \circ \sigma$.

\paragraph{Propositions.}
The universe $\typrop$ is the constant family $\gamma \mapsto \cst \Lo_\infty$. Thus a term of type $\typrop$ is a function $P : |\Gamma| \to \Lo_\infty$, and
\begin{equation*}
\El(P) \defto \gamma \mapsto \decode{P(\gamma)}\;.
\end{equation*}
Hence a term of type $\El(P)$ exists iff $P(\gamma) \le \nu(\gamma)$ for every $\gamma$. In particular, a proposition in the empty context has a closed term iff it has value $0$.
The inclusion of $\typrop$ into $\tytype$ sends $P$ to the same family,
\begin{equation*}
\iota(P) \defto \gamma \mapsto \decode{P(\gamma)}\;,
\qquad
\El(\iota(P)) = \El(P)\;.
\end{equation*}

\paragraph{$\Pi$-types.}
For a context $\Gamma$, a type $A$ over $\Gamma$ and a type $B$ over $\Gamma.A$, write $B_\gamma$ for the family $a \mapsto B(\gamma, a)$ over $A(\gamma)$, and define
\begin{equation*}
\Pi(A, B)(\gamma) \defto \SPi(A(\gamma), B_\gamma)\;.
\end{equation*}
Since the definition is pointwise in $\gamma$, the type former $\Pi$ commutes with substitution.

\begin{lem}\label{lem:pi}
For a type $A$ over $\Gamma$ and a type $B$ over $\Gamma.A$, there is a bijection, commuting with substitution,
\begin{equation*}
\lambda : \{\text{terms of type $B$ over $\Gamma.A$}\}  \rightleftarrows \{\text{terms of type $\Pi(A, B)$ over $\Gamma$}\} :  \mathsf{app}\;.
\end{equation*}
The two inverse laws are the $\beta$- and $\eta$-rules.
\end{lem}
\begin{proof}
A term of type $B$ over $\Gamma.A$ is a function $(\gamma, a) \mapsto y \in B(\gamma, a)$ with
\begin{equation*}
\nu(y) \le \max(\nu(\gamma), \nu(a))\;.
\end{equation*}
By currying, these are the functions assigning to each $\gamma$ an element of $\SPi(A(\gamma), B_\gamma)$ available at $\nu(\gamma)$, that is, the terms of type $\Pi(A, B)$.
The maps $\lambda$ and $\mathsf{app}$ are currying and uncurrying, which are mutually inverse.
Both maps commute with substitution because they are defined pointwise in $\gamma$.
\end{proof}

\paragraph{Impredicative quantification.}
For a type $A$ over $\Gamma$ and $P$ of type $\typrop$ over $\Gamma.A$, let
\begin{equation*}
\forall(A, P)(\gamma) \defto \nu\big(\SPi(A(\gamma),\, \decode{P(\gamma, -)})\big)\;,
\end{equation*}
a term of type $\typrop$ over $\Gamma$.

The staged set $\SPi(A(\gamma), \decode{P(\gamma, -)})$ has at most one element, born at $\forall(A, P)(\gamma)$ by definition.
Hence it is isomorphic to $\decode{\forall(A, P)(\gamma)}$, naturally in $\gamma$, and the types $\Pi(A, \El(P))$ and $\El(\forall(A, P))$ are isomorphic.
The introduction and elimination rules are $\lambda$ and $\mathsf{app}$ of \cref{lem:pi}, composed with this isomorphism.
Since every $\El(P)$ has at most one element, these rules satisfy the substitution laws and the $\beta$- and $\eta$-rules.

\begin{rem}
Inside the type theory, the rules of $\forall$ give an isomorphism
\begin{equation*}
\Pi(A, \El(P)) \cong \El(\forall(A, P))\;.
\end{equation*}
The model can also be modified to validate this isomorphism as a judgmental equality when $\Pi$ is defined by a case distinction, decided in the classical metatheory:
\begin{equation*}
\Pi(A, B)(\gamma) \defto
\begin{cases}
\decode{\nu\big(\SPi(A(\gamma), B_\gamma)\big)} & \text{if $B(\gamma, a)$ has at most one element for every $a \in |A(\gamma)|$}\;,\\
\SPi(A(\gamma), B_\gamma) & \text{otherwise}\;.
\end{cases}
\end{equation*}
The new $\Pi$ still commutes with substitution because the condition is pointwise in $\gamma$.
Since the new $\Pi$ replaces the staged product by an isomorphic staged set at each $\gamma$, \cref{lem:pi} carries over.
Nonetheless, we use the staged product because the description is simpler.

A similar technique can be found in~\cite[\S 4]{MW03}, where the distinction is made syntactically by the sort of $B$, and in an argument of Shulman~\cite[Appendix~A]{Coq23}, where it is made structural by tagging the types that come from propositions.
\end{rem}

\paragraph{$\Sigma$-types and finite types.}
We define $\Sigma(A, B)(\gamma)$ as the staged sum $\SSigma(A(\gamma), B_\gamma)$.
The finite types $\mathbf{0}$, $\mathbf{1}$, $\mathbf{2}$ are the constant families $\cst\emptyset$, $\cst\{\ast\}$ and $\cst\{0, 1\}$.

\paragraph{Identity propositions.}
For a type $A$ over $\Gamma$ and terms $a, b$ of type $A$, the identity proposition $\Id_A(a, b)$ is the term of type $\typrop$
\begin{equation*}
\Id_A(a, b)(\gamma) \defto
\begin{cases}
0 & \text{if } a(\gamma) = b(\gamma)\;,\\
\infty & \text{otherwise}\;.
\end{cases}
\end{equation*}
The term $\refl_a$ of type $\El(\Id_A(a, a))$ is $\gamma \mapsto \ast$.
We write $\Id(a, b)$ when $A$ is clear.

A term of type $\El(\Id_A(a, b))$ exists only if $a = b$.
Thus the model validates equality reflection.

\paragraph{W-types.}
For a type $A$ over $\Gamma$ and a type $B$ over $\Gamma.A$, let $\W(A, B)(\gamma)$ consist of the well-founded trees $\Wsup(a, f)$, with $a \in |A(\gamma)|$ and $f \in \SPi(B(\gamma, a), \W(A, B)(\gamma))$, born at
\begin{equation*}
\nu(\Wsup(a, f)) \defto \max(\nu(a), \nu(f))\;.
\end{equation*}
This is a definition by well-founded recursion on the trees of the set-theoretic W-type of the underlying sets.
Since $\Wsup(a, f)$ is born no later than the point $(\gamma, a, f)$, the constructor $\Wsup$ is a term of type $\W(A, B)$ over the context $\Gamma, a \of A, f \of \pit{b\of{}B\,a} \W(A, B)$.

For a type $C$ over $\Gamma.\W(A, B)$ and a term $c$ of type $C(\Wsup(a, f))$ over the context $\Gamma, a \of A,\allowbreak\ f \of \pit{b\of{}B\,a} \W(A, B),\allowbreak\ g \of \pit{b\of{}B\,a} C(f\,b)$, the recursor is defined by well-founded recursion,
\begin{equation*}
\rec(\gamma, \Wsup(a, f)) \defto c\big(\gamma, a, f, b \mapsto \rec(\gamma, f(b))\big)\;.
\end{equation*}
An induction on trees shows $\nu(\rec(\gamma, t)) \le \max(\nu(\gamma), \nu(t))$, and hence $\rec$ is a term of type $C$.
The recursor commutes with substitution because the definition is pointwise in $\gamma$.
The computation rule holds by definition.

\paragraph{Codes.}
The codes in $\tytype$ of $\Pi$-types, $\Sigma$-types, finite types and W-types are the constructions themselves.
They are $\kappa$-small because $V_\kappa$ is a Grothendieck universe.
The code of $\typrop$ in $\tytype$ is the constant function with value $\cst\Lo_\infty$, which is $\kappa$-small because $\Lo \in V_\kappa$.
The required equations hold by definition.

\subsection{Universes classifying types}\label{sec:why-precomposition}
Following Hofmann and Streicher~\cite{HS97}, a universe classifies the types if decoding is a bijection, natural in $\Gamma$, between the codes $\Gamma \vdash X : \tytype$ and the $\kappa$-small types over $\Gamma$.
The universe $\tytype$ of the model classifies the $\kappa$-small types in this sense, although it is constant, unlike the universe constructed in~\cite{HS97}, whose values vary with the stage.
Indeed, a $\kappa$-small type over $\Gamma$ and a code in $\tytype$ over $\Gamma$ are the same function $|\Gamma| \to \sset_\kappa$, with $\El$ the identity.
Hence $\tytype$ is a Russell universe~\cite[Definition~21]{Kov22}, presented à la Tarski, and $\typrop$ classifies the propositions through the bijection $p \mapsto \decode{p}$.

The standard interpretation in the locally cartesian closed category of staged sets, however, lacks this property.
There, a type over $\Gamma$ is a map into $\Gamma$ and substitution is pullback.
A type over $\Gamma$ is then a family whose fiber over $\gamma$ has no elements born before $\nu(\gamma)$.
Pulling back along a morphism $\sigma : \Delta \to \Gamma$ raises the birthdays of the elements of the fiber over a point $\delta$ to at least $\nu(\delta)$.
The universe $\tytype$ can still be constant, to include the constant universe $\typrop$ that the refutations of excluded middle and propositional extensionality require.
In this case, $\El$ raises the birthdays of a code to the stage of the point.
It then classifies the types only weakly.
That is, every type is the decoding of a code, but no choice of codes is natural in $\Gamma$.

\section{The logic of the model}\label{sec:logic}
For a proposition $P$ and a code $X$ in $\tytype$, we write $P$ for $\El P$ and $X$ for $\El X$, and $P \lto Q$ for $\all{x\of{}P} Q$.
Impredicativity encodes the logical connectives as follows, see for example~\cite[\S 3.1]{AFS18}.
\begin{align*}
\top &\defto \all{R} R \lto R\;,\\
\bot &\defto \all{R} R\;,\\
P \wedge Q &\defto \all{R} (P \lto Q \lto R) \lto R\;,\\
P \vee Q &\defto \all{R} (P \lto R) \lto (Q \lto R) \lto R\;,\\
\ex{x\of{}A} P\,x &\defto \all{R} \big(\all{x\of{}A} (P\,x \lto R)\big) \lto R\;,
\end{align*}
together with $\neg P \defto P \lto \bot$ and $P \liff Q \defto (P \lto Q) \wedge (Q \lto P)$.
Each connective is the impredicative encoding of the corresponding inductive proposition.
For example, the propositional truncation $\supp{X}$ of \cref{sec:intro} encodes the inductive proposition $\mathsf{Nonempty}\,X$ of Lean, with one constructor $X \to \mathsf{Nonempty}\,X$.

The values of these encodings can be computed from the definition of $\forall$.

\begin{lem}\label{lem:values}
For truth values $p, q$, a staged set $X$ and $\varphi : |X| \to \Lo_\infty$,
\begin{align*}
p \lto q &=
\begin{cases}
0 & \text{if } q \le p\;,\\
q & \text{otherwise}\;,
\end{cases} &
\all{x\of{}X} \varphi\,x &= \sup_{x \in X}\,\big(\nu(x) \lto \varphi(x)\big)\;.
\end{align*}
In particular, quantification over $\typrop$, whose elements are all born at $0$, is a supremum: $\all{r\of{}\typrop} F(r) = \sup_{r \in \Lo_\infty} F(r)$ for every $F : \Lo_\infty \to \Lo_\infty$.
Hence the encodings take the values
\begin{align*}
p \wedge q &= \max(p, q)\;, &
p \vee q &= \min(p, q)\;, &
\bot &= \infty\;,\\
\neg p &=
\begin{cases}
0 & \text{if } p = \infty\;,\\
\infty & \text{otherwise}\;,
\end{cases} &
\supp{X} &= \nu(X)\;, &
\ex{x\of{}X} \varphi\,x &= \inf_{x \in X}\,\big(\max(\nu(x), \varphi(x))\big)\;,
\end{align*}
with $\sup \emptyset = 0$ and $\inf \emptyset = \infty$.
\end{lem}

\begin{proof}
By definition, $\all{x\of{}X} \varphi\,x$ is the birthday of $\SPi(X, \decode{\varphi(-)})$: the least $\ell \in \Lo$ with $\varphi(x) \le \max(\ell, \nu(x))$ for all $x \in X$, and $\infty$ if there is none.
For $p \lto q = \all{x\of{}\decode{p}} q$, this is the least $\ell$ with $q \le \max(\ell, p)$, which is $0$ if $q \le p$ and $q$ otherwise.
Since $\varphi(x) \le \max(\ell, \nu(x))$ iff $\nu(x) \lto \varphi(x) \le \ell$, the general case is the stated supremum.
The values of the encodings follow by computing the suprema over $r \in \Lo_\infty$.
For example, for $\supp{X}$ with $m = \nu(X)$, the summand $\big(\all{x\of{}X} r\big) \lto r$ is $r$ if $r \le m$ and $0$ otherwise, and the supremum is $m$.
\end{proof}

\subsection{Propositional extensionality}\label{sec:propext}
For propositions $p$ and $q$, the identity proposition $\Id(p, q)$ is decided at stage $0$, while the equivalence $p \liff q$ may hold only from a later stage on.
In the context $p \of \typrop,\ q \of \typrop$, by \cref{lem:values},
\begin{equation*}
(p \lto q) \wedge (q \lto p)
=
\begin{cases}
0 & \text{if } p = q\;,\\
\max(p, q) & \text{otherwise}\;,
\end{cases}
\end{equation*}
while $\Id(p, q)$ is $0$ if $p$ and $q$ coincide and $\infty$ otherwise.
Hence,
\begin{equation*}
(p \lto q) \wedge (q \lto p) \lto \Id(p, q)
=
\begin{cases}
\infty & \text{if } p \neq q \text{ and } p, q \in \Lo\;,\\
0 & \text{otherwise}\;.
\end{cases}
\end{equation*}
When $\Lo$ has at least two elements, the supremum over $p, q$ is $\infty$, and hence
\begin{equation*}
\neg\,\all{p,q\of{}\typrop} (p \lto q) \wedge (q \lto p) \lto \Id(p, q) = 0\;.
\end{equation*}

\subsection{Excluded middle}\label{sec:lem}
By \cref{lem:values},
\begin{equation*}
\all{p\of{}\typrop} p \vee \neg p = \sup_{p \in \Lo_\infty} (p \vee \neg p) = \sup \Lo\;.
\end{equation*}
When $\Lo$ has no greatest element, in particular when $\Lo = \mathbb{N}$, this is $\infty$, and
\begin{equation*}
\neg\,\all{p\of{}\typrop} p \vee \neg p = 0\;.
\end{equation*}
When $\Lo$ has a greatest element, excluded middle takes the value $\max \Lo$, which is $0$ if $\Lo$ has one element.
Excluded middle is thus independent of the type theory of \cref{sec:intro} with global choice.

\subsection{Global choice}\label{sec:choice}
For each nonempty $\kappa$-small staged set $X$, choose in the metatheory an element $c(X) \in X$ of least birthday. By \cref{lem:values},
\begin{equation*}
\nu(c(X)) = \supp{X}\;.
\end{equation*}

By \cref{lem:pi}, applied twice, a closed term of type
\begin{equation*}
\pit{X\of{}\tytype} \supp{X} \to X
\end{equation*}
is the same as a term of type $X$ in the context $X \of \tytype,\ h \of \supp{X}$.
A point of this context is a pair $(X, \ast)$ with $X$ nonempty, born at $\max(0, \supp{X}) = \supp{X}$, since every code is born at $0$.
Such a term is thus a function $t$ with $t(X, \ast) \in X$ and $\nu(t(X, \ast)) \le \supp{X}$, and
\begin{equation*}
t(X, \ast) \defto c(X)
\end{equation*}
defines one.
The corresponding closed term interprets $\choice$ and, being closed, commutes with substitution.

\subsection{Proof of \texorpdfstring{\cref{thm:main}}{Theorem 1.1}}
Let $\Lo = \mathbb{N}$.
By \cref{lem:values}, $\bot = \infty$, and hence $\bot$ has no closed term.
The model is thus consistent.
It interprets the type theory of \cref{sec:intro} and global choice by \cref{sec:model,sec:choice}.
The negations of propositional extensionality and of excluded middle have closed terms by \cref{sec:propext,sec:lem}.

\section{Further features and Lean}\label{sec:more}
The model of \cref{sec:model} with $\Lo = \mathbb{N}$ validates the following rules beyond those of the type theory of \cref{sec:intro}:
\begin{enumerate}
\item proof irrelevance;
\item a cumulative hierarchy of universes $\tytype_0 \subseteq \tytype_1 \subseteq \cdots$ with codes for the type formers, the inclusion of $\typrop$ into $\tytype_0$, and a global choice operator at every level;
\item inductive families of types and of propositions;
\item quotients with soundness.
\end{enumerate}
Proof irrelevance holds because every $\El(P)$ has at most one element.

\subsection{Universes}
Fix inaccessible cardinals $\kappa_0 < \kappa_1 < \cdots$ with $\Lo \in V_{\kappa_0}$.
The universe $\tytype_i$ is the constant family $\gamma \mapsto \cst\sset_{\kappa_i}$, with $\El$ the identity, as $\tytype$ in \cref{sec:model} with $\kappa_i$ in place of $\kappa$.
This is the standard construction of a cumulative hierarchy of universes in a set-theoretic model.
The code $u_i$ of $\tytype_i$ in $\tytype_{i+1}$ is the constant function with value $\cst\sset_{\kappa_i}$, and the code of $\typrop$ in $\tytype_0$ is the constant function with value $\cst\Lo_\infty$.
For $i \le j$, the lift $\lift_{i,j}$ from $\tytype_i$ to $\tytype_j$ is given by the inclusion $\sset_{\kappa_i} \subseteq \sset_{\kappa_j}$.

Each universe $\tytype_i$ classifies the $\kappa_i$-small types in the sense of \cref{sec:why-precomposition}, with $\El$ the identity.
The construction of \cref{sec:choice} applies to each $\tytype_i$ with $\kappa_i$ in place of $\kappa$, and gives a global choice operator $\choice_i$ at every level.

\subsection{Inductive families}\label{sec:inductive}
Inductive types are encoded with W-types~\cite{Dyb97}, and inductive families reduce to W-types using function extensionality and uniqueness of identity proofs~\cite[\S 7.1]{AGHMM15}.
The computation rules hold by equality reflection.

Inductive families of propositions are encoded impredicatively, as in~\cite{PPM89,AFS18}, and proof irrelevance gives dependent elimination into $\typrop$ from the non-dependent one.

An inductive proposition with at most one constructor, whose arguments are propositions or indices of the result, also eliminates into types.
This is the subsingleton elimination of Lean~\cite[\S 2.6.2]{Carn19}, which extends singleton elimination~\cite{GCST19}.
Global choice derives it in the type theory.
For such a proposition $P$, let $T$ be the inductive family of types obtained by replacing $P$ with $T$ in the type of the constructor.
Then $\supp{T}$ is equivalent to $P$, using global choice.
The eliminator is then given by that of $T$ and global choice.
The type $T$ is a subsingleton, by induction on $T$ using proof irrelevance and function extensionality.
Hence the computation rule holds judgmentally, by equality reflection.

\subsection{Quotients}
Quotients are as in Lean. For $R : A \to A \to \typrop$ there is a type $\Quot(A, R)$ with $\mk : A \to \Quot(A, R)$, lifting, induction and soundness,
\begin{equation*}
\all{a,b\of{}A} R\,a\,b \lto \Id(\mk\,a, \mk\,b)\;.
\end{equation*}
Lean derives exactness,
\begin{equation*}
\all{a,b\of{}A} \Id(\mk\,a, \mk\,b) \lto R\,a\,b
\end{equation*}
for an equivalence relation $R$, from propositional extensionality.
The model does not validate exactness, which with global choice would give excluded middle~\cite{Mai99}.

For a type $A$ over $\Gamma$ and $R$ of type $\typrop$ over $\Gamma.A.A$, let $\approx_\gamma$ be the equivalence relation on $|A(\gamma)|$ generated by $R(\gamma, a, b) \neq \infty$, and
\begin{equation*}
\Quot(A, R)(\gamma) \defto A(\gamma)/{\approx_\gamma}\;,
\qquad
\nu([a]) \defto \min\{\nu(a') \mid a' \approx_\gamma a\}\;.
\end{equation*}
Since a class is born no later than any member, $\mk\,a \defto \gamma \mapsto [a(\gamma)]$ is a term.

For a type $B$ over $\Gamma$, a term $f$ of type $B$ over $\Gamma.A$ and a term $h$ of type $\all{a,b\of{}A} R\,a\,b \lto \Id(f\,a, f\,b)$, lifting is given by $(\qlift\,f\,h)(\gamma, [a]) \defto f(\gamma, a)$.
It is well defined since identity propositions take only the values $0$ and $\infty$.
It is a term since every class has a member born with it, and induction holds for the same reason.

Soundness holds by the definition of $\approx_\gamma$.

\subsection{Lean}
Carneiro's thesis~\cite{Carn19} specifies the type theory of Lean.
The type theory of \cref{sec:intro} with the rules of \cref{sec:more} covers the type theory of Carneiro's thesis except propositional extensionality.
The universes, proof irrelevance, quotients and global choice of Lean are those of \cref{sec:more}, with $\mathsf{Sort}\,0$ as $\typrop$.
The inductive families of Lean, with elimination of subsingleton propositions into types, are given by \cref{sec:inductive}.
The current version of Lean, 4.33.1, also provides mutual and nested inductive types and $\eta$ for structures.
The rule $\eta$ for structures holds in Carneiro's type theory up to propositional equality~\cite[\S 6.1]{Carn25}, and hence judgmentally by equality reflection.
We expect most mutual and nested inductive types to be derivable as well.
Lean reduces nested inductive types to mutual ones~\cite[\S 6]{Carn25}, and finitary mutual inductive families reduce to indexed W-types~\cite[Theorem~9]{KvR20}.

\section{Related work}\label{sec:related}
Goodman and Myhill~\cite{GM78} derive excluded middle from choice in intuitionistic set theory, also using extensionality.
Later work isolates the extensionality principle that the derivation needs.
Bell~\cite{Bel08} works in a weak intuitionistic set theory without extensionality, where excluded middle follows from choice with extensionality for subsets of $2$ or with quotient sets.
In first-order logic, Bell~\cite[\S 3]{Bel93} shows that Hilbert's $\varepsilon$-operator gives excluded middle when it takes formulas with the same extension to the same witness.
In type theory, choice gives excluded middle when extensionality is assumed: of the choice function~\cite{ML06, Car04}, of propositions~\cite{RocqDiaconescu, LeanClassical}, \cite[Theorem~10.1.14]{HoTT13}, or of quotients, as exactness~\cite{Mai99}.

Earlier counter-models with $\varepsilon$-operators interpret first-order or higher-order logic, but not dependent types.
In first-order logic, Bell~\cite{Bel93} constructs a counter-model to excluded middle for the plain $\varepsilon$-operator, with inversely well-ordered truth values~\cite[\S 7]{Bel93}, of which the truth values of the model of \cref{sec:model} are an instance.
In higher-order logic, Maietti, Pasquali and Rosolini~\cite[Example~5.14, Theorem~6.6]{MPR17} use the same truth values: the localic tripos of a well-order with a greatest element has $\varepsilon$-operators without validating excluded middle.
The topos presented by this tripos interprets dependent types, but it has no $\varepsilon$-operators.
The model of \cref{sec:model} interprets dependent type theory with global choice, the guarded form of the $\varepsilon$-operator.

In type theory, the only model known to us with a choice axiom in which excluded middle fails is the intuitionistic set-theoretic model of Sato and Garrigue~\cite{SG23} for $\mathrm{CC}^\omega$, which also refutes propositional extensionality~\cite[\S 4.4]{SG23}.
The type theory $\mathrm{CC}^\omega$ lacks general inductive types, $\Sigma$-types and the inclusion of $\typrop$ into the universes.
In contrast, the core type theories of Lean and Rocq have inductive types and internalize propositions as types, Rocq by the inclusion and Lean through $\mathsf{PLift}$.
Sato and Garrigue interpret a proposition as an open set of an Alexandroff space and a type as a plain set.
Hence a proposition included in a universe forgets this structure, and the interpretation breaks~\cite[\S 6]{SG23}.
In the model of \cref{sec:model}, types and propositions are both staged sets, and a proposition is a staged set with at most one element.
The inclusion of $\typrop$ into the universes is then the identity on staged sets.

\section*{Acknowledgment}
\noindent The author thanks Andrew Swan for discussions and for guidance on how to formally verify a model of type theory, and Andrej Bauer for suggestions on the presentation of this paper.

\noindent
This work was supported by the AI For Math Fund, a program of Renaissance Philanthropy.
It was also supported by the Slovenian Research and Innovation Agency award P1-0294.

\section*{Software use}
Claude Opus 5.5 (Anthropic) was used in developing the model and in formalizing it in Rocq (\cref{app:proofs}).

\bibliographystyle{alphaurl}
\bibliography{references}

\appendix
\crefalias{section}{appendix}  

\section{Proof obligations}\label{app:proofs}
\sloppy  
We checked in Rocq~9.1.1~\cite{Rocq}, with the standard library, that the CwF of \cref{sec:model} satisfies the rules of the type theory of \cref{sec:intro}, together with proof irrelevance and quotients from \cref{sec:more}.
The file {\ttfamily Model.v}, available at \url{https://github.com/park-sewon/Diaconescu/tree/arxiv}, defines the structure \githublink{154}{CwF} of a CwF, and the structure \githublink{1085}{TModel} of the rules on a CwF.
It then defines the CwF \githublink{1536}{stagedCwF} of staged sets and, for any well-ordered $\Lo$ without a greatest element, an instance \githublink{2376}{stagedModel} of \githublink{1085}{TModel} on it, which is instantiated at $\Lo = \mathbb{N}$ as \githublink{2402}{natModel}.
It uses four axioms, all of the classical metatheory: excluded middle, propositional extensionality, dependent function extensionality, and global choice.

Rocq's universes play the role of the Grothendieck universe $V_\kappa$.

Below, $\Ty(\Gamma)$ is the set of types over $\Gamma$, and $\Tm(\Gamma, A)$ is the set of terms of type $A$ over $\Gamma$.
The projection and the variable of the context extension (\cref{sec:model}) are $\pproj : \Gamma.A \to \Gamma$ and $\qvar \in \Tm(\Gamma.A, A[\pproj])$.
For a context morphism $\sigma : \Delta \to \Gamma$, we write $A[\sigma]$ and $t[\sigma]$ for substitution along $\sigma$, and $\sigma^+ \defto \langle \sigma \circ \pproj, \qvar \rangle : \Delta.A[\sigma] \to \Gamma.A$ for the extension of $\sigma$.
For a type $B$ over $\Gamma.A$ and $t \in \Tm(\Gamma, A)$, we write $B[t]$ for $B[\langle \mathrm{id}, t\rangle]$.

Throughout, $A \in \Ty(\Gamma)$ and $B \in \Ty(\Gamma.A)$.
Each rule is an operation on types, an operation on terms, or an equation.
Within an item, a variable keeps the range given in the sentence that introduces it.
An operation on types, such as ``$\Pi(A, B) \in \Ty(\Gamma)$'', is where we define the type $\Pi(A, B)$ over $\Gamma$ for all $A$ and $B$.
Similarly, an operation on terms, such as ``$\lambda(t) \in \Tm(\Gamma, \Pi(A, B))$ for $t \in \Tm(\Gamma.A, B)$'', is where we define the term $\lambda(t)$.
An equation, such as ``$\mathsf{app}(\lambda(t)) = t$'', is what we verify for these definitions.
We further verify that every operation is \emph{natural}: it commutes with substitution along every $\sigma : \Delta \to \Gamma$.
For types, this is an equation such as $\Pi(A, B)[\sigma] = \Pi(A[\sigma], B[\sigma^+])$.
For terms, the two sides have types that are equal only by such an equation, and naturality is stated with transport along it.
For example, $\lambda(t)[\sigma]$ has type $\Pi(A, B)[\sigma]$, and transported to $\Pi(A[\sigma], B[\sigma^+])$ it equals $\lambda(t[\sigma^+])$.

The substitution laws of the introduction and elimination rules of $\forall$ and of $\refl$ are derived in the specification from proof irrelevance.
The $\beta$- and $\eta$-rules of $\forall$ are instances of proof irrelevance.
The eliminator $J$ and function extensionality follow from equality reflection.
We did not formalize these derivations, the universe hierarchy, or the inductive families of \cref{sec:more}.

\subsection*{Category with families}
\begin{enumerate}
\item Substitution: Composition of context morphisms is associative, with identities $\mathrm{id}$. For $t \in \Tm(\Gamma, A)$, $A[\mathrm{id}] = A$ and $t[\mathrm{id}] = t$. For $\tau : \Theta \to \Delta$, $A[\sigma][\tau] = A[\sigma \circ \tau]$ and $t[\sigma][\tau] = t[\sigma \circ \tau]$.\par\hfill{\footnotesize Specification: \githublink{154}{CwF} (\githublink{178}{ty\_sub\_id}, \githublink{180}{ty\_sub\_comp}, \githublink{183}{tm\_sub\_id}, \githublink{185}{tm\_sub\_comp}, \githublink{161}{id\_comp}, \githublink{163}{comp\_id}, \githublink{165}{comp\_assoc}); model: \githublink{1536}{stagedCwF}}
\item Comprehension: For $t \in \Tm(\Delta, A[\sigma])$, $\langle \sigma, t\rangle : \Delta \to \Gamma.A$. $\pproj \circ \langle \sigma, t \rangle = \sigma$. $\qvar[\langle \sigma, t \rangle] = t$. For $\tau : \Delta \to \Gamma.A$, $\langle \pproj \circ \tau, \qvar[\tau] \rangle = \tau$. The empty context is terminal.\par\hfill{\footnotesize Specification: \githublink{154}{CwF} (\githublink{194}{pair\_proj}, \githublink{197}{pair\_var}, \githublink{201}{pair\_eta}, \githublink{170}{toEmpty\_unique}); model: \githublink{1525}{sPair\_proj}, \githublink{1528}{sPair\_var}, \githublink{1531}{sPair\_eta}, \githublink{1518}{sToEmpty\_unique}}
\end{enumerate}

\subsection*{Universe}
\begin{enumerate}\setcounter{enumi}{2}
\item $\tytype \in \Ty(\Gamma)$. For $X \in \Tm(\Gamma, \tytype)$, $\El(X) \in \Ty(\Gamma)$.\par\hfill{\footnotesize Specification: \githublink{993}{UnivBase} (\githublink{995}{U}, \githublink{997}{U\_sub}, \githublink{998}{El}, \githublink{1000}{El\_sub}), \githublink{1085}{TModel} (\githublink{1093}{tU}); model: \githublink{2136}{stUnivB}}
\end{enumerate}

\subsection*{Propositions}
\begin{enumerate}\setcounter{enumi}{3}
\item Propositions: $\typrop \in \Ty(\Gamma)$. For $P \in \Tm(\Gamma, \typrop)$, $\El(P) \in \Ty(\Gamma)$. $\overline{\typrop} \in \Tm(\Gamma, \tytype)$. $\El(\overline{\typrop}) = \typrop$. For $P \in \Tm(\Gamma, \typrop)$, $\iota(P) \in \Tm(\Gamma, \tytype)$. $\El(\iota(P)) = \El(P)$.\par\hfill{\footnotesize Specification: \githublink{722}{PropF} (\githublink{724}{prop}, \githublink{726}{prop\_sub}, \githublink{728}{prfT}, \githublink{730}{prf\_sub}), \githublink{1028}{Codes} (\githublink{1074}{propCode}, \githublink{1075}{El\_propCode}, \githublink{1076}{propCode\_sub}, \githublink{1069}{inclCode}, \githublink{1070}{El\_inclCode}, \githublink{1071}{inclCode\_sub}); model: \githublink{1718}{stProp}, \githublink{2140}{stCodes}}
\item\label{ob:prf-irrel} Proof irrelevance: For $P \in \Tm(\Gamma, \typrop)$ and $h, h' \in \Tm(\Gamma, \El(P))$, $h = h'$.\par\hfill{\footnotesize Specification: \githublink{722}{PropF} (\githublink{733}{prf\_irrel})}
\item Impredicativity: For $P \in \Tm(\Gamma.A, \typrop)$, $\forall(A, P) \in \Tm(\Gamma, \typrop)$. For $b \in \Tm(\Gamma.A, \El(P))$, $\mathsf{intro}(b) \in \Tm(\Gamma, \El(\forall(A, P)))$. For $s \in \Tm(\Gamma, \El(\forall(A, P)))$, $\mathsf{elim}(s) \in \Tm(\Gamma.A, \El(P))$.\par\hfill{\footnotesize Specification: \githublink{722}{PropF} (\githublink{735}{all}, \githublink{737}{all\_sub}, \githublink{741}{allIntro}, \githublink{744}{allElim}), \githublink{753}{allIntro\_sub}, \githublink{758}{allElim\_sub}; model: \githublink{1695}{sAll}, \githublink{1700}{sAllIntro}, \githublink{1710}{sAllElim}}
\end{enumerate}

\subsection*{$\Pi$ and $\Sigma$}
\begin{enumerate}\setcounter{enumi}{6}
\item $\Pi$-types: $\Pi(A, B) \in \Ty(\Gamma)$. For $t \in \Tm(\Gamma.A, B)$, $\lambda(t) \in \Tm(\Gamma, \Pi(A, B))$. For $s \in \Tm(\Gamma, \Pi(A, B))$, $\mathsf{app}(s) \in \Tm(\Gamma.A, B)$. $\mathsf{app}(\lambda(t)) = t$. $\lambda(\mathsf{app}(s)) = s$.\par\hfill{\footnotesize Specification: \githublink{313}{PiF} (\githublink{317}{pi\_sub}, \githublink{320}{lam}, \githublink{321}{app}, \githublink{323}{pi\_beta}, \githublink{325}{pi\_eta}, \githublink{327}{lam\_sub}); model: \githublink{1586}{stPi}, \githublink{1317}{PiS}, \githublink{1575}{sLam}, \githublink{1580}{sApp}, \githublink{1326}{piExt}}
\item $\Sigma$-types: $\Sigma(A, B) \in \Ty(\Gamma)$. For $t \in \Tm(\Gamma, A)$ and $u \in \Tm(\Gamma, B[t])$, $(t, u) \in \Tm(\Gamma, \Sigma(A, B))$. For $s \in \Tm(\Gamma, \Sigma(A, B))$, $\mathsf{fst}(s) \in \Tm(\Gamma, A)$ and $\mathsf{snd}(s) \in \Tm(\Gamma, B[\mathsf{fst}(s)])$. $\mathsf{fst}(t, u) = t$. $\mathsf{snd}(t, u) = u$. $(\mathsf{fst}(s), \mathsf{snd}(s)) = s$.\par\hfill{\footnotesize Specification: \githublink{331}{SigmaF} (\githublink{335}{sigma\_sub}, \githublink{338}{spair}, \githublink{340}{sfst}, \githublink{341}{ssnd}, \githublink{344}{sigma\_fst}, \githublink{347}{sigma\_snd}, \githublink{350}{sigma\_eta}, \githublink{353}{spair\_sub}, \githublink{359}{sfst\_sub}, \githublink{362}{ssnd\_sub}); model: \githublink{1617}{stSigma}}
\end{enumerate}

\subsection*{Finite types}
\begin{enumerate}\setcounter{enumi}{8}
\item Empty type: $\mathbf{0} \in \Ty(\Gamma)$. For $C \in \Ty(\Gamma.\mathbf{0})$ and $e \in \Tm(\Gamma, \mathbf{0})$, $\mathsf{elim}_C(e) \in \Tm(\Gamma, C[e])$.\par\hfill{\footnotesize Specification: \githublink{370}{FinF} (\githublink{372}{zero}, \githublink{374}{zero\_sub}, \githublink{376}{zero\_elim}, \githublink{378}{zero\_elim\_sub}); model: \githublink{1660}{stFin}}
\item Unit type: $\mathbf{1} \in \Ty(\Gamma)$. $\ast \in \Tm(\Gamma, \mathbf{1})$. For $e \in \Tm(\Gamma, \mathbf{1})$, $e = \ast$. The eliminator of $\mathbf{1}$ and the computation rule of the eliminator follow from this equation.\par\hfill{\footnotesize Specification: \githublink{370}{FinF} (\githublink{388}{one}, \githublink{390}{one\_sub}, \githublink{391}{star}, \githublink{393}{star\_sub}, \githublink{395}{one\_eta}), \githublink{442}{one\_elim}, \githublink{446}{one\_elim\_star}; model: \githublink{1660}{stFin}}
\item Booleans: $\mathbf{2} \in \Ty(\Gamma)$. $\mathsf{tt}, \mathsf{ff} \in \Tm(\Gamma, \mathbf{2})$. For $C \in \Ty(\Gamma.\mathbf{2})$, $c \in \Tm(\Gamma, C[\mathsf{tt}])$, $d \in \Tm(\Gamma, C[\mathsf{ff}])$ and $e \in \Tm(\Gamma, \mathbf{2})$, $\mathsf{if}_C(c, d, e) \in \Tm(\Gamma, C[e])$. $\mathsf{if}_C(c, d, \mathsf{tt}) = c$. $\mathsf{if}_C(c, d, \mathsf{ff}) = d$.\par\hfill{\footnotesize Specification: \githublink{370}{FinF} (\githublink{397}{two}, \githublink{399}{two\_sub}, \githublink{400}{ttT}, \githublink{401}{ffT}, \githublink{403}{tt\_sub}, \githublink{405}{ff\_sub}, \githublink{407}{two\_elim}, \githublink{411}{two\_elim\_tt}, \githublink{414}{two\_elim\_ff}, \githublink{417}{two\_elim\_sub}); model: \githublink{1660}{stFin}, \githublink{1645}{sTwoElim}}
\end{enumerate}

\subsection*{Identity propositions}
\begin{enumerate}\setcounter{enumi}{11}
\item Identity propositions: For $a, b \in \Tm(\Gamma, A)$, $\Id_A(a, b) \in \Tm(\Gamma, \typrop)$. For $a \in \Tm(\Gamma, A)$, $\refl_a \in \Tm(\Gamma, \El(\Id_A(a, a)))$.\par\hfill{\footnotesize Specification: \githublink{798}{IdF} (\githublink{800}{idP}, \githublink{802}{idP\_sub}, \githublink{805}{refl}), \githublink{811}{refl\_sub}; model: \githublink{1734}{sIdP}, \githublink{1742}{sRefl}}
\item Reflection: For $a, b \in \Tm(\Gamma, A)$, $a = b$ if $\Tm(\Gamma, \El(\Id_A(a, b)))$ is nonempty.\par\hfill{\footnotesize Specification: \githublink{798}{IdF} (\githublink{807}{reflection}); model: \githublink{1747}{sReflection}}
\end{enumerate}

\subsection*{W-types}
\begin{enumerate}\setcounter{enumi}{13}
\item $\W(A, B) \in \Ty(\Gamma)$.

For $F \defto \Pi(B, \W(A, B)[\pproj \circ \pproj]) \in \Ty(\Gamma.A)$, $\Wsup \in \Tm(\Gamma.A.F, \W(A, B)[\pproj \circ \pproj])$.
Here $\pproj \circ \pproj$ is the double projection, which makes $\W(A, B)[\pproj \circ \pproj]$ a type over $\Gamma.A.B$ and over $\Gamma.A.F$.

For $C \in \Ty(\Gamma.\W(A, B))$ and $c \in \Tm(\Gamma, M)$, where
\begin{equation*}
M \defto \pit{a\of{}A} \pit{f\of{}\pit{b\of{}B\,a} \W(A, B)} \big(\pit{b\of{}B\,a} C(f\,b)\big) \to C(\Wsup(a, f))\;,
\end{equation*}
$\rec(C, c) \in \Tm(\Gamma.\W(A, B), C)$.

The computation rule is
\begin{equation*}
\rec(C, c)(\Wsup(a, f)) = c\,a\,f\,(\lambda b.\,\rec(C, c)(f\,b))\;.
\end{equation*}
\par\hfill{\footnotesize Specification: \githublink{676}{WF} (\githublink{678}{wty}, \githublink{680}{wty\_sub}, \githublink{683}{sup}, \githublink{685}{sup\_sub}, \githublink{696}{wrec}, \githublink{699}{wrec\_sup}, \githublink{708}{wrec\_sub}); model: \githublink{1763}{sW}, \githublink{1773}{supVal}, \githublink{1792}{sSupT}, \githublink{1893}{wrecF}, \githublink{1848}{wNodeAt}, \githublink{1905}{wrecF\_node}, \githublink{1928}{wrecF\_sub}}
\end{enumerate}

\subsection*{Quotients}
\begin{enumerate}\setcounter{enumi}{14}
\item Quotients: For $R \in \Tm(\Gamma.A.A[\pproj], \typrop)$, $\Quot(A, R) \in \Ty(\Gamma)$, with $\Quot(A, R)[\sigma] = \Quot(A[\sigma], R[(\sigma^+)^+])$. $\mk \in \Tm(\Gamma.A, \Quot(A, R)[\pproj])$.\par\hfill{\footnotesize Specification: \githublink{837}{QuotF} (\githublink{839}{quot}, \githublink{841}{quot\_sub}, \githublink{844}{qmk}, \githublink{847}{qmk\_sub}); model: \githublink{2029}{sQuot}, \githublink{2105}{sQuot\_sub}, \githublink{2037}{sQmk}, \githublink{1388}{QS}, \githublink{1389}{cls}, \githublink{1407}{rep}}
\item Lifting: For $C \in \Ty(\Gamma)$, $f \in \Tm(\Gamma.A, C[\pproj])$ and a term $h$ of
\begin{equation*}
\El\big(\all{a,b\of{}A} R\,a\,b \lto \Id(f\,a, f\,b)\big)\;,
\end{equation*}
$\qlift(f, h) \in \Tm(\Gamma.\Quot(A, R), C[\pproj])$. $\qlift(f, h)[\langle \pproj, \mk\rangle] = f$.\par\hfill{\footnotesize Specification: \githublink{837}{QuotF} (\githublink{853}{qlift}, \githublink{857}{qlift\_mk}, \githublink{864}{qlift\_sub}), \githublink{828}{respIn}; model: \githublink{2060}{sQlift}, \githublink{2042}{resp\_eqc}}
\item Induction: For $P \in \Tm(\Gamma.\Quot(A, R), \typrop)$ and $e \in \Tm(\Gamma.A, \El(P[\langle \pproj, \mk\rangle]))$, $\mathsf{ind}(e) \in \Tm(\Gamma.\Quot(A, R), \El(P))$.\par\hfill{\footnotesize Specification: \githublink{837}{QuotF} (\githublink{876}{qind}), \githublink{967}{qind\_sub}; model: \githublink{2068}{sQind}}
\item Soundness: $\mathsf{sound} \in \Tm\big(\Gamma, \El\big(\all{a,b\of{}A} R\,a\,b \lto \Id(\mk\,a, \mk\,b)\big)\big)$.\par\hfill{\footnotesize Specification: \githublink{837}{QuotF} (\githublink{880}{qsound}), \githublink{828}{respIn}, \githublink{980}{qsound\_sub}; model: \githublink{2080}{sQsound}}
\end{enumerate}

\subsection*{Codes}
\begin{enumerate}\setcounter{enumi}{18}
\item For each type former, a code in $\tytype$ whose $\El$ is the former applied to the decodings. For example, for $X \in \Tm(\Gamma, \tytype)$ and $Y \in \Tm(\Gamma.\El(X), \tytype)$, $\pi(X, Y) \in \Tm(\Gamma, \tytype)$. $\El(\pi(X, Y)) = \Pi(\El(X), \El(Y))$. The formers are $\Pi$, $\Sigma$, $\mathbf{0}$, $\mathbf{1}$, $\mathbf{2}$, W-types and quotients.\par\hfill{\footnotesize Specification: \githublink{1028}{Codes} (\githublink{1031}{piCode}, \githublink{1032}{El\_piCode}, \githublink{1034}{piCode\_sub}, \githublink{1037}{sigmaCode}, \githublink{1043}{zeroCode}, \githublink{1047}{oneCode}, \githublink{1051}{twoCode}, \githublink{1055}{wCode}, \githublink{1061}{quotCode}), \githublink{1085}{TModel} (\githublink{1094}{tCodes}); model: \githublink{2140}{stCodes}}
\end{enumerate}

\subsection*{Choice, excluded middle and propositional extensionality}
\begin{enumerate}\setcounter{enumi}{19}
\item Global choice: A closed term $\choice$ of $\pit{X\of{}\tytype} \El(\supp{X}) \to \El(X)$.\par\hfill{\footnotesize Specification: \githublink{1085}{TModel} (\githublink{1096}{choiceT}), \githublink{1020}{truncT}; model: \githublink{2218}{sChoice}, \githublink{2176}{trunc\_le}, \githublink{2204}{chooseAt}}
\item Refutation of excluded middle: A closed term of $\El\big(\neg\,\all{P\of{}\typrop} P \vee \neg P\big)$.\par\hfill{\footnotesize Specification: \githublink{1085}{TModel} (\githublink{1101}{notEM}), \githublink{786}{EMT}, \githublink{778}{negT}; model: \githublink{2357}{sNotEM}, \githublink{2333}{em\_not\_le}}
\item Refutation of propositional extensionality: A closed term of $\El\big(\neg\,\all{P,Q\of{}\typrop} (P \liff Q) \lto \Id(P, Q)\big)$.\par\hfill{\footnotesize Specification: \githublink{1085}{TModel} (\githublink{1103}{notPropext}), \githublink{821}{propextT}; model: \githublink{2352}{sNotPropext}, \githublink{2303}{propext\_not\_le}, \githublink{2367}{stageNZ}}
\end{enumerate}

\subsection*{Consistency}
There is no closed term of $\El(\bot)$.\par\hfill{\footnotesize Specification: \githublink{1085}{TModel} (\githublink{1105}{consistent}), \githublink{777}{botT}; model: \githublink{2363}{sConsistent}, \githublink{2243}{bot\_none}}

The lemmas \githublink{2429}{sEl\_mkCode}, \githublink{2435}{mkCode\_sEl} and \githublink{2439}{mkCode\_sub} state that the universe $\tytype$ classifies the $\kappa$-small types in the sense of \cref{sec:why-precomposition}: decoding is a bijection between codes and $\kappa$-small types, natural in the context.
The lemma \githublink{2445}{univ\_classifies} packages the three.

\end{document}